\documentclass[11pt]{article}

\usepackage{tabularx,booktabs,multirow,delarray,array}
\usepackage{graphicx,amssymb,amsmath,amssymb}
\usepackage{enumerate}
\usepackage{algorithmic}
\usepackage[ruled]{algorithm2e}
\usepackage{wrapfig}
\usepackage{latexsym}
\usepackage{lineno}
\usepackage{lscape}
\usepackage{fullpage}
\usepackage{bibspacing}
\newcommand{\ignore}[1]{ }

\newcounter{rem}
\def\calP{\mathcal{P}}

\def\calF{\mathcal{F}}

\def\qed{\hbox{\rlap{$\sqcap$}$\sqcup$}}
\newcommand{\Tri}{\mbox{$T\!r\!i$}}

\def\qed{\hbox{\rlap{$\sqcap$}$\sqcup$}}
\newtheorem{theorem}{Theorem}[section]
\newtheorem{lemma}{Lemma}[section]

\newenvironment{proof}{\par\noindent{\bf Proof:}}{\mbox{}\hfill$\qed$\\}
\newtheorem{observation}{Observation}

\begin{document}

\title{Approximate Euclidean shortest paths in polygonal domains} 
\author{
R. Inkulu\footnote{Department of Computer Science \& Engg, IIT Guwahati.  E-mail: rinkulu@iitg.ac.in}
\and
Sanjiv Kapoor\footnote{Department of Computer Science, Illinois Institute of Technology, Chicago, USA.  E-mail: kapoor@iit.edu}
}
\date{}
\maketitle

\maketitle

\begin{abstract}
Given a set $\mathcal{P}$ of $h$ pairwise disjoint simple polygonal obstacles in $\mathbb{R}^2$ defined with $n$ vertices, we compute a sketch $\Omega$ of $\mathcal{P}$ whose size is independent of $n$, depending only on $h$ and the input parameter $\epsilon$. 
We utilize $\Omega$ to compute a $(1+\epsilon)$-approximate geodesic shortest path between the two given points in $O(n + h((\lg{n}) + (\lg{h})^{1+\delta} + (\frac{1}{\epsilon}\lg{\frac{h}{\epsilon}})))$ time.
Here, $\epsilon$ is a user parameter, and $\delta$ is a small positive constant (resulting from the time for triangulating the free space of $\cal P$ using the algorithm in \cite{journals/ijcga/Bar-YehudaC94}).
Moreover, we devise a $(2+\epsilon)$-approximation algorithm to answer two-point Euclidean distance queries for the case of convex polygonal obstacles.
\end{abstract}

\section{Introduction}
\label{sect:intro}

For any set $\cal{Q}$ of pairwise-disjoint simple polygonal obstacles in $\mathbb{R}^2$, the free space $\mathcal{F(Q)}$ is the closure of $\mathbb{R}^2$ without the union of the interior of all the polygons in $\cal{Q}$.
Given a set $\mathcal{P} = \{P_1, P_2, \ldots, P_h\}$ of pairwise-disjoint simple polygonal obstacles in $\mathbb{R}^2$ and two points $s$ and $t$ in $\mathcal{F(P)}$, the {\it Euclidean shortest path finding problem} seeks to compute a shortest path between $s$ and $t$ that lies in $\mathcal{F(P)}$. 
This problem is well-known in the computational geometry community.
Mitchell~\cite{books/compgeom/TothRG17,coll/hb/Mitch00} provides an extensive survey of research accomplished in determining shortest paths in polygonal and polyhedral domains.
The problem of finding shortest paths in graphs is quite popular and considered to be fundamental. 
Especially, several algorithms for efficiently computing single-source shortest paths and all-pairs shortest paths are presented in Cormen et al.~\cite{books/algo/Cormen09} and Kleinberg and Tardos ~\cite{books/algo/KleinTard05}) texts.
And, the algorithms for approximate shortest paths are surveyed in \cite{conf/walcom/SSen09}.
In the following, we assume that $n$ vertices together define the $h$ polygonal obstacles of $\mathcal{P}$.

Given a polygonal domain $\mathcal{P}$ as input, the following are three well-known variants of the Euclidean shortest path finding problem: (i) both $s$ and $t$ are given as input with $\mathcal{P}$, (ii) only $s$ is provided as input with $\mathcal{P}$, and (iii) neither $s$ nor $t$ is given as input.
The type (i) problem is a single-shot problem and involves no preprocessing. 
The preprocessing phase of the algorithm for a type (ii) problem constructs a shortest path map with $s$ as the source so that a shortest path between $s$ and any given query point $t$ can be found efficiently.
In the third variation, which is known as a two-point shortest path query problem, the polygonal domain $\mathcal{P}$ is preprocessed to construct data structures that facilitate in answering shortest path queries between any given pair of query points $s$ and $t$.

In solving a type (i) or type (ii) problem, there are two fundamentally different approaches: the visibility graph method (see Ghosh~\cite{books/visalgo/skghosh2007} for both the survey and details of various visibility algorithms) and the continuous Dijkstra (wavefront propagation) method.
The visibility graph method \cite{journals/talg/ChenW15,journals/siamcomp/KapoorM00,journals/dcg/KapoorMM97,journals/ipl/Welzl85} is based on constructing a graph $G$, termed visibility graph, whose nodes are the vertices of the obstacles (together with $s$ and $t$) and edges are the pairs of mutually visible vertices.
Once the visibility graph $G$ is available, a shortest path between $s$ and $t$ in $G$ is found using Dijkstra's algorithm.
As the number of edges in the visibility graph is $O(n^2)$, this method has worst-case quadratic time complexity.
In the continuous Dijkstra approach \cite{journals/siamcomp/HershbergerS99,journals/corr/InkuluKM10,conf/stoc/Kapoor99,journals/ijcga/Mitchell96}, a wavefront is expanded from $s$ till it reaches $t$.
In specific, for the case of polygonal obstacles in plane, Hershberger and Suri devised an algorithm in \cite{journals/siamcomp/HershbergerS99} which computes a shortest path in $O(n\lg{n})$ time and the algorithm in \cite{journals/corr/InkuluKM10} (which extends the algorithm by Kapoor \cite{conf/stoc/Kapoor99}) by Inkulu, Kapoor, and Maheshwari computes a shortest path in $O(n+h((\lg{h})^{\delta}+(\lg{n})(\lg{h})))$ time.
Here, $\delta$ is a small positive constant (resulting from the time for triangulating the $\cal{F(P)}$ using the algorithm in \cite{journals/ijcga/Bar-YehudaC94}). 
The continuous Dijkstra method typically constructs a shortest path map with respect to $s$ so that for any query point $t$, a shortest path from $s$ to $t$ can be found efficiently. 

The two-point shortest path query problem within a given simple polygon was addressed by Guibas and Hershberger~\cite{journals/jcss/GuibasH89}.
It preprocessed the simple polygon in $O(n)$ time and constructed a data structure of size $O(n)$ and answers two-point shortest distance queries in $O(\lg{n})$ time.
Exact two-point shortest path queries in the polygonal domain were explored by Chiang and Mitchell~\cite{conf/soda/ChiangM99}.
One of the algorithms in \cite{conf/soda/ChiangM99} constructs data structures of size $O(n^5)$ and answers the query for any two-point distance in $O(h + \lg{n})$ worst-case time.
And, another algorithm in \cite{conf/soda/ChiangM99} builds data structures of size $O(n+h^5)$ and outputs any two-point distance query in $O(h\lg{n})$ time.
In both of these algorithms, a shortest path itself is found in additional time $O(k)$, where $k$ is the number of edges in the output path.
Guo et~al.~\cite{conf/aaim/GuoMS08} preprocessed $\cal{F(P)}$ in $O(n^2\lg{n})$ time to compute data structures of size $O(n^2)$ for answering two-point distance queries for any given pair of query points in $O(h\lg{n})$ time.

Because of the difficulty of exact two-point queries in polygonal domains, various approximation algorithms were devised.
Clarkson first made such an attempt in \cite{conf/stoc/Clarkson87}.
Chen~\cite{conf/soda/Chen95} used the techniques from \cite{conf/stoc/Clarkson87} in constructing data structures of size $O(n\lg{n} + \frac{n}{\epsilon})$ in $o(n^{3/2})+O(\frac{n}{\epsilon}\lg{n})$ time to support $(6+\epsilon)$-approximate two-point distance queries in $O(\frac{1}{\epsilon}\lg{n}+\frac{1}{\epsilon^2})$ time, and a shortest path in additional $O(L)$ time, where $L$ is the number of edges of the output path.
Arikati et~al.~\cite{conf/esa/ArikatiCCDSZ96} devised a family of algorithms to answer two-point approximate shortest path queries. 
Their first algorithm outputs a $(\sqrt{2}+\epsilon)$-approximate distance; depending on a parameter $1 \le r \le n$, in the worst-case, either the preprocessed data structures of this algorithm take $O(n^2)$ space or the query time is $O(\sqrt{n})$. 
Their second algorithm takes $O(n)$ query time to report the distance.
The stretch of the third and fourth algorithms proposed in \cite{conf/esa/ArikatiCCDSZ96} are respectively  $(2\sqrt{2}+\epsilon)$ and $(3\sqrt{2}+\epsilon)$.
Agarwal et~al.~\cite{conf/soda/AgarwalSY09} computes a $(1+\epsilon)$-approximate geodesic shortest path in $O(n + \frac{h}{\sqrt{\epsilon}}\lg(\frac{h}{\epsilon}))$ time when the obstacles are convex.

Throughout this paper, to distinguish graph vertices from the vertices of the polygonal domain, we refer to vertices of a graph as nodes.
The Euclidean distance between any two points $p$ and $q$ is denoted with $\Vert pq \Vert$.
The obstacle-avoiding geodesic Euclidean shortest path distance between any two points $p, q$ amid a set $\cal{Q}$ of obstacles is denoted with $dist_{\mathcal{Q}}(p, q)$.
The (shortest) distance between two nodes $s$ and $t$ in a graph $G$ is denoted with $dist_G(s, t)$.
Unless specified otherwise, distance is measured in Euclidean metric.
We denote both the convex hull of a set $R$ of points and the convex hull of a simple polygon $R$ with $CH(R)$.
Let $r'$ and $r''$ be two rays with origin at $p$.
Let $\overrightarrow{v_1}$ and $\overrightarrow{v_2}$ be the unit vectors along the rays $r'$ and $r''$ respectively. 
A {\it cone} $C_p(r', r'')$ is the set of points defined by rays $r'$ and $r''$ such that a point $q \in C_p(r', r'')$ if and only if $q$ can be expressed as a convex combination of the vectors $\overrightarrow{v_1}$ and $\overrightarrow{v_2}$ with positive coefficients.
When the rays are evident from the context, we denote the cone with $C_p$.
The counterclockwise angle from the positive x-axis to the line that bisects the cone angle of $C_p$ is termed as the {\it orientation of the cone $C_p$}.

\subsection*{Our contributions}
\label{subsect:contrib}

First, we describe the algorithm for the case in which $\calP$ comprises convex polygonal obstacles.
We compute a sketch $\Omega$ from the polygonal domain $\calP$.
Essentially, each convex polygonal obstacle $P_i$ in $\calP$ is approximated with another convex polygonal obstacle whose complexity depends only on the input parameter $\epsilon$; significantly, the size of the approximated polygon is independent of the size of $P_i$.
In specific, when $\calP$ is comprised of $h$ convex polygonal obstacles, the sketch $\Omega$ is comprised of $h$ convex polygonal obstacles: for each $1 \le i \le h$, the convex polygon $P_i \in \cal{P}$ is approximated with another convex polygon $Q_i \in \Omega$.
For each $P_i \in \cal{P}$, we identify a coreset $S_i$ of vertices of $P_i$ and form the core-polygon $Q_i \in \Omega$ using $S_i$.
When $P_i$ is convex, the corresponding core-polygon $Q_i$ obtained through this procedure is convex; and, $Q_i \subseteq P_i$.
Like in \cite{conf/soda/AgarwalSY09}, the combinatorial complexity of $\Omega$ is independent of $n$; it depends only on $h$ and the input parameter $\epsilon$.
For two points $s, t \in \cal{F(P)}$, we compute an approximate Euclidean shortest path between $s$ and $t$ in ${\cal F}(\Omega)$ using an algorithm that is a variant of \cite{conf/stoc/Clarkson87}.
From this path, we compute a path $R$ in $\cal{F(P)}$ and show that $R$ is a $(1+\epsilon)$-approximate Euclidean shortest path between $s$ and $t$ amid polygonal obstacles in $\cal{P}$.
When the obstacles in $\calP$ are not necessarily convex, we compute the sketch of $\calP$ using the convex chains (that bound the obstacles) as well as the corridor paths that result from the hourglass decomposition \cite{conf/stoc/Kapoor99,conf/socg/Kapoor88,journals/dcg/KapoorMM97} of $\cal{F(P)}$.
The main contributions and the major advantages in our approach are described in the following:
\begin{itemize}
\item
When $\cal{P}$ is comprised of disjoint simple polygonal obstacles, we compute a $(1+\epsilon)$-approximate geodesic Euclidean shortest path between the two given points belonging to $\cal{F(P)}$ in $O(n + h((\lg{n}) + (\lg{h})^{1+\delta} + \frac{1}{\epsilon}\lg{\frac{h}{\epsilon}}))$ time.
Here, $\delta$ is a small positive constant resulting from the triangulation of the free space using the algorithm from \cite{journals/ijcga/Bar-YehudaC94}.
(Refer to Theorem~\ref{thm:sppolydom}.)
Agarwal et~al.~\cite{conf/soda/AgarwalSY09} compute a $(1+\epsilon)$-approximate geodesic shortest path in $O(n + \frac{h}{\sqrt{\epsilon}}\lg(\frac{h}{\epsilon}))$ time when the obstacles are convex.
In computing approximate shortest paths, our algorithm extends the notion of coresets in \cite{conf/soda/AgarwalSY09} to simple polygons.
However, our approach is computing coresets, and an approximate shortest path using these coresets is quite different from \cite{conf/soda/AgarwalSY09}.
Our algorithm to construct the sketch of $\calP$ is simpler.

\item
As part of devising the above algorithm, when $\cal{P}$ is comprised of convex polygonal obstacles, our algorithm computes a $(1+\epsilon)$-approximate geodesic Euclidean distance between the two given points in $O(n+ \frac{h}{\epsilon}\lg{\frac{h}{\epsilon}})$ time.
Further, our algorithm computes a $(1+\epsilon)$-approximate shortest path in additional $O(h\lg{n})$ time.
(Refer to Theorem~\ref{thm:spcvx}.)

\item
When $\cal{P}$ is comprised of disjoint convex polygonal obstacles, we preprocess these polygons in $O(n+\frac{h}{\epsilon^2}(\lg{\frac{h}{\epsilon}})+\frac{h}{\epsilon}(\lg{\frac{h}{\epsilon}})^2)$ time to construct data structures of size $O(\frac{h}{\epsilon})$ for answering any two-point $(2+\epsilon)$-approximate geodesic distance (length) query in $O(\frac{1}{\epsilon^6}(\lg{\frac{h}{\epsilon}})^2)$ time.
(Refer to Theorem~\ref{thm:distq}.)
To compute an optimal geodesic shortest path amid simple polygonal obstacles, Chen and Wang~\cite{journals/talg/ChenW15} takes $O(n+h\lg{h}+k)$ time, where $k$ is a parameter sensitive to the geometric structures of the input and is upper bounded by $O(h^2)$.
Our algorithm to answer approximate two-point distance queries amid convex polygonal obstacles takes space close to linear in $n$ whereas the preprocessed data structures of algorithms proposed in \cite{conf/soda/ChiangM99} occupy $\Omega(n^5)$ space in the worst-case.
Also, our algorithm for two-point distance queries improves the stretch factor of \cite{conf/soda/Chen95} from $(6+\epsilon)$ to $(2+\epsilon)$ in case of convex polygonal obstacles.

\item
Furthermore, our algorithm to compute the coreset of simple polygons to obtain a sketch $\Omega$ of $\cal{P}$ as well as the algorithm to compute an approximate geodesic Euclidean shortest path in $\calP$ using the sketch $\Omega$ may be of independent interest.
\end{itemize}

Section~\ref{sect:r2nopreproc} describes an algorithm for computing a single-shot approximate shortest path when obstacles in $\calP$ are convex polygons.
Section~\ref{sect:polydom} extends this algorithm to compute an approximate Euclidean shortest path amid simple polygonal obstacles. 
The algorithm to answer two-point approximate Euclidean distance queries amid convex polygonal obstacles is described in Section~\ref{sect:r2preproc}.
A table comparing earlier algorithms to ours is given in the Appendix.

\section{Approximate shortest path amid convex polygons}
\label{sect:r2nopreproc}

In this section, we consider the case in which every simple polygon in $\cal{P}$ is convex. 
We use the following notation from Yao~\cite{journals/siamcomp/Yao82}.
Let $\kappa \ge 2$, and define $\theta = 2\pi/\kappa$.
Consider the set of $\kappa$ rays: for $0 \le i \kappa$, the ray $r_i$ passes through the origin and makes an angle $i\theta$ with the positive $x$-axis.
Each pair of successive rays defines a cone whose apex is at the origin. 
This collection of $\kappa$ cones is denoted by $\mathcal{C}$.
It is clear that the cones of $\mathcal{C}$ partition the plane.
Also, the two bounding rays of any cone of $\mathcal{C}$ make an angle $\theta$.
In our algorithm, the value of $\kappa$ is chosen as a function of $\epsilon$ (refer to Subsection~\ref{subsect:gcon}).
When a cone $C \in \mathcal{C}$ is translated to have the apex at a point $p$, the translated cone is denoted with $C_p$.
Each cone that we refer in this paper is a translated copy of some cone in $\mathcal{C}$.
For each polygon $P$ in $\mathcal{P}$, we choose a subset of $O(\frac{1}{\sqrt{\alpha\epsilon}})$ vertices from the vertices of $P$.  
At each such vertex $p$, we introduce a set of cones at $p$.
In the algorithm to compute a single-shot $s$-$t$ geodesic shortest path, the value of $\alpha$ is set to $\frac{\epsilon}{2}$.
The algorithm for two-point approximate distance queries sets the value of $\alpha$ to $\frac{\epsilon}{12}$. 
The proof of Theorem~\ref{thm:spcvx} details the reasons for setting these specific values.
As detailed below, these vertices and cones help in computing a spanner that approximates a Euclidean shortest path between the two given points in $\mathcal{F(P)}$.

\subsection{Sketch of $\mathcal{P}$}

In this subsection, we define and characterise the sketch of $\cal{P}$. 
For any $P_i \in \calP$ and any two points $p'$ and $p''$ on the boundary of $P_i$, the section of boundary of $P_i$ that occurs while traversing from $p'$ to $p''$ in counterclockwise order is termed a {\it patch} of $P_i$. 
In specific, we partition the boundary of each $P_i \in \cal{P}$ into a collection of patches $\Gamma_i$ such that for any two points $p', p''$ belonging to any patch $\gamma \in \Gamma_i$, the angle between the outward (w.r.t. the centre of $P_i$) normals to respective edges at $p'$ and $p''$ is upper bounded by $\frac{\sqrt{\alpha\epsilon}}{2}$.
The maximum angle between the outward normals to any two edges that belong to a patch $\gamma$ constructed in our algorithm is the {\it angle subtended by $\gamma$}.
To facilitate in computing patches of any obstacle $P_i$, we partition the unit circle $\mathbb{S}^2$ centred at the origin into a minimum number of segments such that each circular segment is of length at most $\frac{\sqrt{\alpha\epsilon}}{2}$.
For every such segment $s$ of $\mathbb{S}^2$, a patch (corresponding to $s$) comprises of the maximal set of the contiguous sequence of edges of $P_i$ whose outward normals intersect $s$, when each of these normals is translated to the origin.
(To avoid degeneracies, we assume each normal intersects a single segment.)
Let $\Gamma_i$ be a partition of the boundary of a convex polygon $P_i$ into a collection of $O(\frac{1}{\sqrt{\alpha\epsilon}})$ patches.
The lemma below shows that the geodesic distance between any two points belonging to any patch $\gamma \in \Gamma_i$ is a $(1+\alpha\epsilon)$-approximation to the Euclidean distance between them.

\begin{lemma}
\label{lem:patchdist}
For any two points $p$ and $q$ that belong to any patch $\gamma \in \Gamma_j$, the geodesic distance between $p$ and $q$ along $\gamma$ is upper bounded by
$(1+\alpha\epsilon)\Vert pq \Vert$ for $\alpha\epsilon < 1$.
\end{lemma}
\begin{proof}
Let $e'$ be the edge on which $p$ lies and let $e''$ be the edge on which $q$ lies.
Let $c$ be the point of intersection of normal to $e'$ at $p$ and the normal to $e''$ at $q$.
\begin{figure}[h]
\begin{minipage}[t]{\linewidth}
\begin{center}
\includegraphics[totalheight=0.9in]{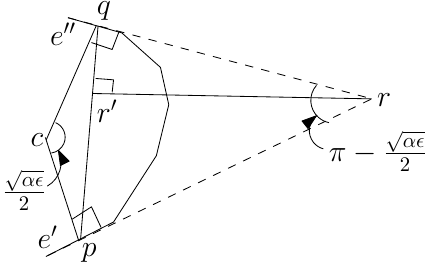}
\end{center}
\vspace{-0.15in}
\caption{\footnotesize Illustrating the construction in proving the upper bound on the patch length}
\label{fig:patchapprx}
\end{minipage}
\end{figure}
Since $p$ and $q$ belong to the same patch, the angle between $cp$ and $cq$ is upper bounded by $\frac{\sqrt{\alpha\epsilon}}{2}$, when the value of $\alpha\epsilon$ is small.
Let $l'$ and $l''$ be the lines that respectively pass through $e'$ and $e''$.
Also, let $r$ be the point at which lines $l'$ and $l''$ intersect.
(Refer to Fig.~\ref{fig:patchapprx}.)
For the small values of $\sqrt{\alpha\epsilon}$ and due to triangle inequality, the geodesic length of the patch between $p$ and $q$ is upper bounded by $\Vert pr \Vert + \Vert qr \Vert$.
Let $r'$ be the point of projection of $r$ on to line segment $pq$.
Suppose $\angle{qrr'} = \angle{prr'}$.
(Analysis of other cases is similar.)
Then $\Vert pr \Vert + \Vert qr \Vert
	\le \frac{\Vert pr' \Vert}{\sin(\frac{\pi}{2} - \frac{\sqrt{\alpha\epsilon}}{4})} + \frac{\Vert r'q \Vert}{\sin(\frac{\pi}{2} - \frac{\sqrt{\alpha\epsilon}}{4})}
	= \frac{\Vert pq \Vert}{\cos{\frac{\sqrt{\alpha\epsilon}}{4}}}
    \le (1+\alpha\epsilon) \Vert pq \Vert
    $.
The last inequality is valid when $\alpha\epsilon < 1$.
\end{proof}

For each obstacle $P_i$, the {\it coreset $S_i$ of $P_i$} is comprised of two vertices chosen from each patch in $\Gamma_i$.
In particular, for each patch $\gamma \in \Gamma_i$, the first and last vertices of $\gamma$ that occur while traversing the boundary of $P_i$ are chosen to be in the coreset $S_i$ of $P_i$. 
The {\it coreset $\cal{S}$ of $\mathcal{P}$} is then simply $\bigcup_i S_i$.

\begin{observation}
\label{obs:coresetsize}
The size of the coreset $\cal{S}$ of $\cal{P}$ is $O(\frac{h}{\sqrt{\alpha\epsilon}})$.
\end{observation}

For every $1 \le i \le h$, our algorithm uses {\it core-polygon} $Q_i = CH(S_i)$ in place of $P_i$.
For any single point obstacle $P_i$ in $\calP$, the core-polygon of $P_i$ is that point itself.
The patch construction procedure guarantees that each polygonal obstacle in $\cal P$ is partitioned into patches such that the core-polygon that correspond to every obstacle in $\cal{P}$ is valid.
Let $\Omega$ be the set comprising of core-polygons corresponding to each of the polygons in $\cal{P}$. 
The set $\Omega$ is called the {\it sketch of $\cal{P}$}.
The following lemmas show that $\Omega$ facilitates in computing a $(1+\alpha\epsilon)$-approximation of the geodesic distance between any two given points in $\cal{F(P)}$.

\begin{lemma}
\label{lem:sketchpt}
Let $v', v''$ be any two vertices of obstacles in $\Omega$.
Then, $dist_\mathcal{P}(v', v'')$ is upper bounded by $(1+\alpha\epsilon) dist_\Omega(v', v'')$.
\end{lemma}

\begin{proof}
Let $v_1, v_2$ be any two successive vertices along a shortest path between $v'$ and $v''$ in ${\calF(\Omega)}$.
Let $\mathcal{O} \subseteq \calP$ be the set of obstacles intersected by the line segment $v_1v_2$.
Let $v_1$ and $v_2$ be respectively belonging to obstacles $P_j$ and $P_k$.
Also, let $\Gamma_j$ (resp. $\Gamma_k$) be the set comprising the partition of boundary of $P_j$ (resp. $P_k$) into patches.
And, let $S_j$ (resp. $S_k$) be the coreset of $P_j$ (resp. $P_k$).
Since the line segment $v_1v_2$ does not intersect the interior of the $CH(S_j)$ or $CH(S_k)$, it intersects at most one patch belonging to set $\Gamma_j$ and at most one patch belonging to set $\Gamma_k$.
Let $v_1$ and $r$ be the points of intersection of line segment $v_1v_2$ with a patch $\gamma \in \Gamma_j$.
(These points might as well be the endpoints of $\gamma$.)
Then from Lemma~\ref{lem:patchdist}, the geodesic distance between $v_1$ and $r$ along $\gamma$ is upper bounded by $(1+\alpha\epsilon)\Vert v_1 r \Vert$.
(Refer Fig.~\ref{fig:sppatchintersect}.)
Analogously, let $v_2$ and $r'$ be the points of intersection of line segment $v_1v_2$ with a patch $\gamma' \in \Gamma_k$.
Then the geodesic distance between $v_2$ and $r'$ is upper bounded by $(1+\alpha\epsilon) \Vert v_2r' \Vert$.
For any convex polygonal obstacle $P_l$ in $\mathcal{O}$ distinct from $P_j$ and $P_k$, let $p', p''$ be the points of intersection of $v_1v_2$ with the boundary of $P_l$.  
Since the line segment $v_1v_2$ does not intersect the interior of the convex hull of coreset corresponding to $P_l$, both $p'$ and $p''$ belong to the same patch, say $\gamma'' \in \Gamma_l$.
Then again from Lemma~\ref{lem:patchdist}, the geodesic distance between $p'$ and $p''$ along patch $\gamma''$ is upper bounded by $(1+\alpha\epsilon) \Vert p'p'' \Vert$.
We modify $v_1v_2$ as follows: For every maximal subsection, say $p_i'p_i''$, of the line segment $v_1v_2$ that is interior to a polygonal obstacle of $\cal P$, we replace that subsection with a geodesic Euclidean shortest path in $\cal{F(P)}$ between $p_i'$ and $p_i''$.

\begin{figure}[h]
\begin{minipage}[t]{\linewidth}
\begin{center}
\includegraphics[totalheight=0.9in]{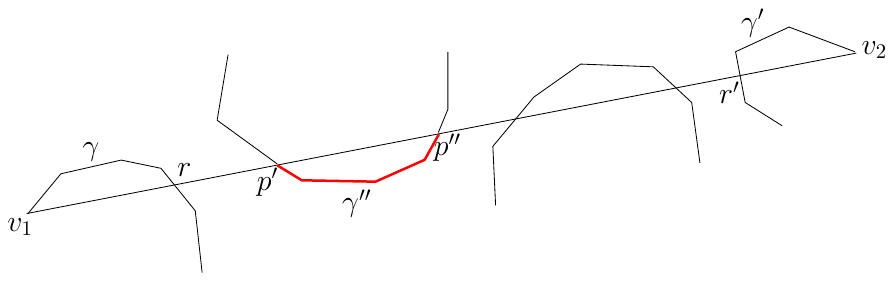}
\vspace{-0.15in}
\caption{\footnotesize A line segment $v_1v_2$ of a shortest path amid $\Omega$ intersecting three patches belonging to obstacles in $\cal{P}$}
\label{fig:sppatchintersect}
\end{center}
\end{minipage}
\end{figure}

Let $\gamma_1, \gamma_2, \ldots, \gamma_k$ be the set $\Gamma$ of patches intersected by the line segment $v_1v_2$. 
Also, for every $1 \le i \le k$, let $p_i', p_i''$ be the points of intersections of $v_1v_2$ with patch $\gamma_i \in \Gamma$ with $p_i'$ closer to $v_1$ than $v_2$ along the line segment $v_1v_2$.
Then $\sum_{i=1}^k dist_{\cal{P}}(p_i', p_i'')$ added with $\sum_{i=1}^{k-1} \Vert p_i''p_{i+1}' \Vert$ is upper bounded by $(1+\alpha\epsilon) \Vert v_1v_2 \Vert$.
Let $v_1, \ldots, v_l$ be the vertices of $\cal{P}$ that occur in that order along a Euclidean shortest path in $\calF(\Omega)$ between vertices $v',v'' \in \cal{P}$.
Then $dist_{\cal{P}}(v_1, v_l) = \sum_{i=1}^{l-1} dist_{\cal{P}}(v_i, v_{i+1}) \le (1+\alpha\epsilon) \sum_{i=1}^{l-1} dist_{\Omega}(v_i, v_{i+1})$.
Note that we do this transformation for each line segment of the shortest path that intersects any patch.
\end{proof}

Since $\mathcal{F(P)} \subseteq \mathcal{F}(\Omega)$, every path that avoids convex polygonal obstacles in $\cal{P}$ is also a path that avoids convex polygonal obstacles in $\Omega$.
This observation leads to the following:

\begin{lemma}
\label{lem:ppt}
For any two vertices $v', v''$ of $\mathcal{P} $, $dist_{\Omega}(v',v'') \leq dist_{\mathcal{P}}(v',v'') $.
\end{lemma}

Considering the given two points $s, t \in \calF(\calP)$ as degenerate obstacles, a $(1+\alpha\epsilon)$-approximation of the shortest distance between $s$ and $t$ amid polygonal obstacles in $\calP$ is computed.

\begin{lemma}
\label{lem:numptsketch}
For a set ${\cal P}$ of $h$ pairwise disjoint convex polygons in $\mathbb{R}^2$ and two points $s, t \in \cal{F(P)}$, the sketch $\cal{S}$ of $\mathcal{P}$ with cardinality $O(\frac{h}{\sqrt{\alpha\epsilon}})$ suffices to compute a $(1+\alpha\epsilon)$-approximate shortest path between $s$ and $t$ in $\mathcal{F(P)}$.
\end{lemma}
\begin{proof}
Immediate from Observation~\ref{obs:coresetsize}, Lemma~\ref{lem:sketchpt}, and Lemma~\ref{lem:ppt}.
\end{proof}

Our approach in computing coresets and an approximate shortest path using these coresets is quite different from \cite{conf/soda/AgarwalSY09}. 
As will be shown in Section~\ref{sect:polydom}, our sketch construction is extended to compute shortest paths even when $\cal{P}$ comprises of polygon obstacles which are not necessarily convex. 
How our algorithm differs from \cite{conf/soda/AgarwalSY09} for the convex polygonal case is detailed herewith.
Let $\calP$ be the polygonal domain defined with convex polygons $P_1, P_2, \ldots, P_h$.
In this algorithm as well as in \cite{conf/soda/AgarwalSY09}, $P_i$ is approximated with $Q_i$, for every $1 \le i \le h$.
However, for every $1 \le i \le h$, in our algorithm $Q_i \subseteq P_i$ whereas in \cite{conf/soda/AgarwalSY09}, $P_i \subseteq Q_i$.
Let the new polygonal domain $\Omega$ be defined with simple polygons $Q_1, Q_2, \ldots, Q_h$.
Unlike \cite{conf/soda/AgarwalSY09}, in computing $\Omega$, our algorithm does not require using plane sweep algorithm to find pairwise vertically visible simple polygons of ${\cal P}$.
As described above, our algorithm partitions the boundary of each convex polygon $P$ into a set of patches.

\subsection{Computing an approximate geodesic shortest path in $\cal{F(P)}$ using the sketch $\Omega$ of $\cal{P}$}
\label{subsect:gcon}

Since we intend to compute an approximate shortest path, to keep our algorithm simpler, we do not want to use the algorithm from \cite{journals/siamcomp/HershbergerS99} to compute a shortest path amid convex polygonal obstacles in $\Omega$, 
Instead, we use a spanner constructed with the conic Voronoi diagrams ($CVD$s) \cite{conf/stoc/Clarkson87}.
Further, in our algorithm, for any maximal line segment with endpoints $r', r''$ along the computed (approximate) shortest path amid obstacles in $\Omega$, if the line segment $r'r''$ lies in $\cal{F(Q)} -  \cal{F(P)}$, we replace line segment $r'r''$ with the geodesic Euclidean shortest path between $r'$ and $r''$ in ${\cal F(P)}$.

Since our algorithm relies on \cite{conf/stoc/Clarkson87}, we give a brief overview of that algorithm first.
The algorithm in \cite{conf/stoc/Clarkson87} constructs a spanner $G(V, E)$ for polygonal domain $\calP$.
Noting that the endpoints of line segments of a shortest path in $\cal{F(P)}$ are a subset of vertices of polygonal obstacles in $\calP$, the node set $V$ is defined as the vertex set of $\mathcal{P}$.
Let $\cal{C}'$ be the set of $O(\frac{1}{\epsilon})$ cones with apex at the origin of the coordinate system together partitioning $\mathbb{R}^2$.
(The cone angle of each cone in $\cal{C}'$ except for one is set to $\epsilon$ and that one cone has $2\pi-\lfloor \frac{2\pi}{\epsilon} \rfloor \epsilon$ as the cone angle.)
Let $C \in \mathcal{C}'$ be a cone with orientation $\theta$ and let $C' \in \mathcal{C}'$ be the cone with orientation $-\theta$.
For each cone $C \in \mathcal{C}'$ and a set $K$ of points, the set of cones resultant from introducing a cone $C_p$ for every point $p \in K$, is the conic Voronoi diagram $CVD(C, K)$.
(Note that as mentioned earlier, $C_p$ is the cone resulted from translating cone $C$ to have the apex at the point $p$.)
For a given cone $C_v$, among all the points on the boundaries of polygons in $\calP$ that are visible from $v$, a point $p$ whose projection onto the bisector of $C_v$ is closest to $v$ is said to be {\it a closest point in $C_v$ to $v$}.
If more than one point is closest in $C_v$ to $v$, then we arbitrarily pick one of those points.
For every vertex $v$ of $\calP$ and for every cone $C_v$, if a closest point $p$ in $C_v-\{v\}$ to $v$ is not a vertex of $\calP$, then the algorithm includes $p$ as a node in $V$.  
Further, for every vertex $v$ of $\mathcal{P}$ and for every cone $C_v$, an edge $e$ joining $v$ and a closest point $p$ in $C_v-\{v\}$ to $v$ is introduced in $E$ with its weight equal to the Euclidean distance between $v$ and $p$.
For every node $v$ in $G$ that corresponds to a point $p$ on the boundary of $P \in \calP$, if $p$ is not a vertex of $\calP$, then for every neighbor $p'$ of $p$ on the boundary of $P$ which has a corresponding node $v'$ in $V$, we introduce an edge $e'$ between $v$ and $p'$ into $E$ and set the weight of $e'$ equal to the Euclidean distance between $v$ and $p'$.
These are the only edges included in $E$.
The Theorem~$2.5$ in \cite{conf/stoc/Clarkson87} proves that if $d$ is the obstacle-avoiding geodesic Euclidean shortest path distance between any two vertices, say $v'$ and $v''$, of $\mathcal{P}$, then the distance between the corresponding nodes $v'$ and $v''$ in $G$ is upper bounded by $(1+\epsilon)d$.
The $CVD(C, K)$ is computed using the plane sweep in $O(|K|\lg{|K|})$ time; and, the well-known planar point location data structure is used to locate the region in $CVD(C, K)$ to which a given query point belongs to. 

As detailed below, apart from computing a sketch $\Omega$ of ${\cal P}$, as compared with \cite{conf/stoc/Clarkson87}, the number of cones per obstacle that participate in computing $CVD$s amid ${\cal F}(\Omega)$ is further optimized by exploiting the convexity of obstacles together with the properties of shortest paths amid convex obstacles.
By limiting the number of vertices of $\mathcal{P}$ at which the cones are initiated to coreset $\cal{S}$ of vertices, our algorithm improves the space complexity of the algorithm in \cite{conf/stoc/Clarkson87}.
Further, by exploiting the convexity of obstacles, we introduce $O(\frac{1}{\sqrt{\alpha\epsilon}})$ cones per obstacle, each with cone angle $O(\sqrt{\alpha\epsilon})$, and show that these are sufficient to achieve the claimed approximation factor. 

Let $v$ be a vertex of $\mathcal{P}$ that belongs to coreset $S_i$ of convex polygon $P_i$. 
Let $v', v, v''$ be the vertices that respectively occur while traversing the boundary of $P_i$ in counterclockwise order.
Also, let $C'$ be the cone defined by the pair of rays $(\overrightarrow{vv'},- \overrightarrow{vv''})$ and let $C''$ be the cone defined by the pair of rays $(\overrightarrow{vv''}, -\overrightarrow{vv'})$.
\begin{wrapfigure}{r}{0.55\textwidth}
\begin{minipage}[t]{\linewidth}
\vspace{-10pt}
\begin{center}
\includegraphics[totalheight=1.3in]{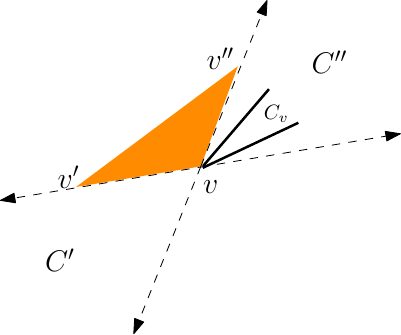}
\end{center}
\vspace{-20pt}
\caption{\footnotesize Illustrating an admissible cone $C_v$ incident to a coreset vertex $v$ of an obstacle}
\label{fig:admisscones}
\vspace{-2pt}
\end{minipage}
\end{wrapfigure}
For a coreset vertex $v \in {\cal S}$, a cone $C \in \mathcal{C}$ is said to be {\it admissible} at $v$ whenever $C_v \cap C'$ or $C_v \cap C''$ is non-empty. 
(See Fig.~\ref{fig:admisscones}.)
Let $p$ and $q$ be two points in $\mathcal{F(P)}$ such that $p$ and $q$ are not visible to each other due to polygonal obstacles in $\calP$.
Let $v$ be a vertex of $P_i$ through which a shortest path between $p$ and $q$ passes.
Since any shortest path is convex at $v$ with respect to $P_i$, there exists a shortest path between $p$ and $q$ where one of its line segment lies in $C'$, and another line segment of that path lies in $C''$. 
Hence, in computing a Euclidean shortest path amid $\cal{P}$, it suffices to consider admissible cones at the vertices of $\cal{P}$.

Note that whenever two points $s$ and $t$ between which we intend to find a shortest path are visible to each other, the line segment $st$ needs to be computed. 
To facilitate this, for every degenerate point obstacle $p$, every cone $C$ with apex $p$ is considered to be an admissible cone.

The same properties carry over to the polygonal domain $\Omega$ as well.
For any two points $p_1$ and $p_2$ in $\mathcal{F}(\Omega)$, suppose that $p_1$ and $p_2$ are not visible to each other.   
Consider any shortest path $\tau$ between $p_1$ and $p_2$.
For any line segment $ab$ in $\tau$, $ab$ is either an edge of a polygon in $\Omega$ or it is a tangent to an obstacle $O \in \Omega$.
In the latter case, $ab$ belongs to an admissible cone of $O$.
When the polygonal domain is $\Omega$, the following Lemma upper bounds the number of cones at the vertices of convex polygons in $\Omega$. 

\begin{lemma}
\label{lem:numcones}
The number of cones introduced at all the obstacles of $\Omega$ is $O(\frac{h}{\sqrt{\alpha\epsilon}})$.
\end{lemma}
\begin{proof}
Let $O$ be the origin of the coordinate system.
Let $\overrightarrow{r}$ be a ray with origin at $O$.
(See Fig.~\ref{fig:numadmcones}.)
For any two distinct vertices $v'$ and $v''$ of a convex polygon $P$, let $\overrightarrow{r_{v'}}$ be the ray parallel to $\overrightarrow{r}$ with origin at $v'$ and pointing in the same direction as $\overrightarrow{r}$ and let $\overrightarrow{r_{v''}}$ be the ray parallel to $\overrightarrow{r}$ with origin at $v''$ and point in the same direction as $\overrightarrow{r}$.
Also, let $v_1'$ precede $v'$ (resp. $v_1''$ precede $v''$) and $v_2'$ succeed $v'$ (resp. $v_2''$ succeed $v''$) while traversing the boundary of $P$ in counterclockwise order.
Since $P$ is a convex polygon, if every point of $\overrightarrow{r_{v'}}$ belongs to the cone defined by $\overrightarrow{v_1'v'}$ and $\overrightarrow{v'v_2'}$ then it is guaranteed that not every point of $\overrightarrow{r_{v''}}$ belongs to the cone defined by $\overrightarrow{v''v_2''}$ and $\overrightarrow{v_1''v''}$.
\begin{figure}[h]
\center{
\includegraphics[totalheight=1.2in]{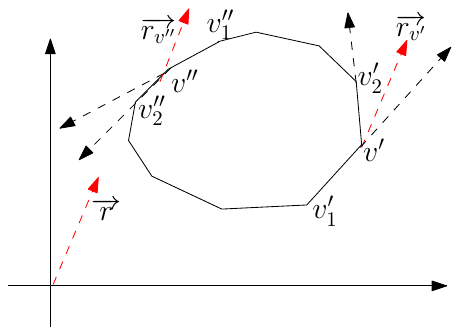}
\caption{\footnotesize Illustrating that a ray parallel to $r$ can exist in only one admissible cone per obstacle} 
\label{fig:numadmcones}
}
\end{figure}
Extending this argument, if a cone $C_{v'}$ is admissible at $v'$ then the cone $C_{v''}$ cannot be admissible at $v''$.
Since the number of coreset vertices per obstacle is $O(\frac{1}{\sqrt{\alpha\epsilon}})$, the number of cones introduced per obstacle is $O(\frac{1}{\sqrt{\alpha\epsilon}})$.
Further, since there are $h$ convex polygonal obstacles, number of cones at all the obstacle vertices together is $O(\frac{h}{\sqrt{\alpha\epsilon}})$.
\end{proof}

Next, we describe the algorithm to compute the spanner $G(V = S \cup S', E)$. 
The set $S$ comprises of nodes corresponding to coreset $\cal{S}$. 
The set $S'$ is a set of Steiner points, as follows.
For every $v \in S$ and every admissible cone $C_v$, let $V'$ be the set of points on the boundaries of obstacles of $\Omega$ that are visible from $v$ and belong to cone $C_v$.
(See Fig.~\ref{fig:coresetspanner}.)
The point $p$ in $V'$ that is closest to $v$, termed the {\it closest Steiner point in $C_v$} to $v$, is determined and $p$ is added to $S'$.
An edge $e$ between $v$ and $p$ is introduced in $E$ while the Euclidean distance between $v$ and $p$ is set as the weight of $e$ in $G$.
Let $p$ be located on a convex polygonal obstacle $P$.
\begin{wrapfigure}{r}{0.5\textwidth}
\begin{minipage}[t]{\linewidth}
\vspace{-10pt}
\begin{center}
\includegraphics[totalheight=0.7in]{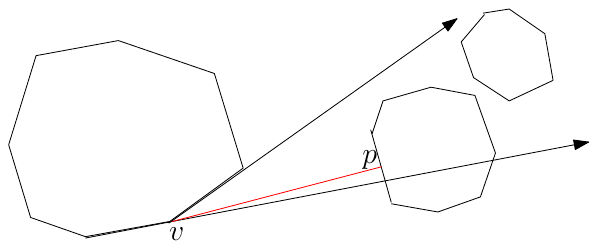}
\centering
\caption{\footnotesize Illustrating an edge of the spanner}
\label{fig:coresetspanner}
\end{center}
\end{minipage}
\end{wrapfigure}
Further, for every Steiner point $p$, let $v'$ (resp. $v''$) be the coreset vertex or Steiner point that lies on the boundary of $P$ and occurs before (resp. after) $p$ while traversing the boundary of $P$ in counterclockwise order.
Then an edge $e'$ (resp. $e''$) between $p$ and $v'$ (resp. $p$ and $v''$) is introduced in $E$ while the geodesic distance between $p$ and $v'$ (resp. $p$ and $v''$) along the boundary of $P$ is set as the weight of $e'$ (resp. $e''$) in $G$. 
Note that both $|V|$ and $|E|$ are $O(\frac{h}{\sqrt{\alpha\epsilon}})$.
For any two points $s, t \in \calF(\Omega)$, the following Lemma upper bounds the $dist_G(s, t)$ in terms of $dist_\Omega(s, t)$.

\begin{lemma}
\label{lem:coneangle}
Let $G$ be the spanner constructed from $\Omega$.
Let $dist_G(p', p'')$ be the distance between $p'$ and $p''$ in $G$.
Then for any two points $s, t \in {\cal F}(\Omega)$, $dist_{\Omega}(s, t) \le dist_G(s, t) \le (1+\sqrt{\alpha\epsilon})dist_{\Omega}(s, t)$.
\end{lemma}
\begin{proof}
Theorem~2.5 of \cite{conf/stoc/Clarkson87} concludes that to achieve $(1+\alpha\epsilon)$-approximation, $\sin{\psi} - \cos{\psi} \le \frac{-1}{1+\alpha\epsilon}$.
Expanding $sine$ and $cosine$ functions for the first few terms yield $-1 + \psi + \frac{\psi^2}{2!} \le \frac{-1}{1+\alpha\epsilon}$.
Solving the quadratic equation in $\psi$ yields $\psi \le \alpha\epsilon$.
Since we are using cones with cone angle $\sqrt{\alpha\epsilon}$ in our algorithm, a $(1+\sqrt{\alpha\epsilon})$-approximation is achieved. 

We claim that introducing a subset of cones (admissible cones) rather than all the cones as used in  \cite{conf/stoc/Clarkson87} does not affect the correctness. 
Let $p$ and $q$ be the vertices of two convex polygons $P_i$ and $P_j$ respectively.
Suppose that $pq$ is a line segment belonging to a shortest path $R$ between vertices $s$ and $t$ of the spanner computed in \cite{conf/stoc/Clarkson87}. 
Further, suppose that $p$ occurs before $q$ when $R$ is traversed from $s$ to $t$. 
If the line along $pq$ supports $P_i$ (resp. $P_j$) at $p$ (resp. $q$), then the line segment $pq$ belongs to an admissible cone at $p$ (resp. $q$).
Otherwise, there exists a line segment in the admissible cone with apex either at a vertex of $P_i$ or at a vertex of $P_j$ which would yield a shorter path from source $s$ to $q$ without using the line segment $pq$.
\end{proof}

Once we find a shortest path $SP_\Omega$ between $s$ and $t$ amid convex polygonal obstacles in $\Omega$ using the spanner $G$, following the proof of Lemma~\ref{lem:ppt}, we transform $SP_\Omega$ to a path amid obstacles in $\mathcal{P}$.
Since there are $O(h)$ obstacles in $\Omega$, $SP_\Omega$ contains $O(h)$ tangents between obstacles.
Let this set of tangents be $\cal{T}$.
We need to find points of intersection of convex polygons in $\cal{P}$ with the line segments in $\cal{T}$.
For any $l \in \cal{T}$ and $P_i \in \cal P$, by using the algorithm from Dobkin et~al.~\cite{journals/tcs/DobkinK83}, we compute the possible intersection between $l$ and $P_i$.
Whenever a line segment $l \in \cal{T}$ and a convex polygon $P_i \in \cal{P}$ intersect, say at points $p'$ and $p''$, we replace the line segment between $p'$ and $p''$ with the geodesic shortest path between $p'$ and $p''$ along the boundary of $P_i$.
Analogously, for every line segment $l \in SP_\Omega - \cal{T}$ belonging to an obstacle $P_j \in \Omega$, we replace $l$ with the corresponding geodesic path along the boundary of $P_j$.
We use the plane sweep technique \cite{books/compgeom/deberg2008} to determine whichever line segments in $\cal{T}$ could intersect with the convex obstacles in $\cal{P}$.
Essentially, the event handling procedures of plane sweep algorithm replace every line segment in $SP_\Omega$ that intersects with any obstacle $P_j \in \calP$ with the shortest geodesic shortest path along the boundary of $P_j$, so that the resulting shortest path $SP_\calP$ after all such replacements belongs to $\cal{F(P)}$.

As part of the plane sweep, a vertical line is swept from left-to-right in the plane.
Let $L$ (resp. $R$) be the set of leftmost (resp. rightmost) vertices of convex polygons in $\cal{P}$.
Initially, points in $L$ and $R$ together with the two endpoints of every line segment in $\cal{T}$ are inserted into the priority queue $Q$.
The event points are scheduled from $Q$ using their respective distances from the initial sweep line position.
As the events occur, the event points corresponding to $L, R$, and the endpoints of line segments in $\cal{T}$ are handled and are deleted from $Q$. 
The algorithm terminates whenever $Q$ is empty.
As described below, the intersection points between the line segments in $\cal{T}$ and the convex polygons in $\cal{P}$ are added to $Q$ with the traversal of the sweep line. 
The sweep line status is maintained as a balanced binary search tree $B$.
We insert (resp. delete) a pointer to a line segment in $\cal{T}$ or a pointer to a convex polygon in $\cal{P}$ to $B$ whenever leftmost (resp. rightmost) endpoint of it is popped from $Q$.
We note that before a line segment $l \in \cal{T}$ and $P \in \calP$ intersect, it is guaranteed that $l$ and $P$ occur adjacent along the sweep line. 
Hence, whenever $l$ and $P$ are adjacent in the sweep line status, we update the event-point schedule with the point of intersection between $l$ and $P$ that occurs first among all such points of intersection in traversing the sweep line from left to right. 
By using the algorithm from Dobkin et~al.~\cite{journals/tcs/DobkinK83}, we compute the possible intersection between $l$ and $P$.
If they do intersect, we push the leftmost point of their intersection to $Q$ with the distance from the initial sweep line as the priority of that event point.
Further, we store the rightmost intersection point between $l$ and $P$ with the leftmost point of intersection as satellite data.
If the leftmost intersection point between $l$ and $P$ pops from $Q$, we compute the geodesic shortest path along the boundary of $P$ between the leftmost intersection point and the corresponding rightmost intersection point.
Further, whenever $l$ and $P$ become non-adjacent along the sweep line, we delete their leftmost point of intersection from $Q$.

\begin{theorem}
\label{thm:spcvx}
Given a set $\mathcal{P}$ of pairwise disjoint convex polygons, two points $s, t \in \mathcal{F(P)}$, and $\epsilon \in (0, 0.6]$, computing a $(1+\epsilon)$-approximate geodesic distance between $s$ and $t$ takes $O(n+ \frac{h}{\epsilon}\lg{\frac{h}{\epsilon}})$ time.
Further, within an additional $O(h\lg{n})$ time, a $(1+\epsilon)$-approximate shortest path is computed.
\end{theorem}
\begin{proof}
From Lemma~\ref{lem:numptsketch}, we know that $dist_{\Omega}(s, t) \le dist_{\cal P}(s, t) \le (1+\alpha\epsilon) dist_{\Omega}(s, t)$.
Let $G$ be the spanner constructed.
From Lemma~\ref{lem:coneangle}, we know that $dist_{\Omega}(s,t) \le dist_{G}(s, t) \le (1+\sqrt{\alpha\epsilon})dist_{\Omega}(s, t)$.
As detailed in Lemma~\ref{lem:sketchpt}, algorithm transforms a shortest path between $s$ and $t$ in $G$ to a path $p$ in $\cal{F(P)}$.
Let $dist_{\calP}^{p}(s, t)$ be the distance along $p$.
From Lemma~\ref{lem:sketchpt}, $dist_{\calP}^{p}(s, t) \le (1+\alpha\epsilon)dist_{G}(s, t)$.
Hence, $dist_{\calP}^{p}(s, t) 
    \le (1+\alpha\epsilon)dist_{G}(s, t)
    \le (1+\alpha\epsilon)(1+\sqrt{\alpha\epsilon})dist_{\Omega}(s, t)
    \le (1+\alpha\epsilon)(1+\sqrt{\alpha\epsilon})dist_{\calP}(s, t)$.
Since $p$ is a path in $\cal{F(P)}$, it is immediate to note that $dist_{\cal P}(s, t) \le dist_{\calP}^{p}(s, t)$.
Therefore, $dist_{\cal P}(s, t) \le dist_{\calP}^{p}(s, t) \le (1+\alpha\epsilon)(1+\sqrt{\alpha\epsilon}) dist_{\cal P}(s, t)$.
To achieve $(1+\epsilon)$-approximation, $(1+\alpha\epsilon)(1+\sqrt{\alpha\epsilon})$ needs to be less than or equal to $(1+\epsilon)$.
For small values of $\epsilon$ ($\epsilon \in (0, 0.6]$), choosing $\alpha = \frac{\epsilon}{2}$ satisfies this inequality.

From here on, we denote $\alpha\epsilon$ with $\epsilon'$.
Finding the coreset $\cal{S}$ of vertices from the convex polygons in $\calP$, and computing the set $\Omega$ of core-polygons together takes $O(n)$ time.
The number of coreset vertices is $O(\frac{h}{\sqrt{\epsilon'}})$.
The number of cones per obstacle is $O(\frac{1}{\sqrt{\epsilon'}})$.
Therefore, the total number of cones is $O(\frac{h}{\sqrt{\epsilon'}})$.
For any cone $C \in \mathcal{C}$ and for any core-polygon $O \in \Omega$, at most a constant number of vertices of $O$ are apexes to cones that have the orientation of $C$.
Considering a sweep line in the orientation of $C$, the sweep line algorithm to find the closest Steiner point to the apex of each cone $C$ (whenever an obstacle intersects with $C$) takes $O(h\lg{h})$ time.
Hence, computing the set of closest Steiner points corresponding to all the cone orientations in $\mathcal{C}$ together take $O(\frac{h}{\sqrt{\epsilon'}}\lg{h})$.

The number of nodes in the spanner $G$ is $O(\frac{h}{\sqrt{\epsilon'}})$.
These nodes include coreset vertices and at most one closest Steiner point per cone.
As each cone introduces at most one edge into $G$, the number of edges in $G$ is $O(\frac{h}{\sqrt{\epsilon'}})$.
Using the Fredman-Tarjan algorithm \cite{journals/jacm/FredmanT87}, finding a shortest path between $s$ and $t$ in $G$ takes $O(\frac{h}{\sqrt{\epsilon'}}\lg{\frac{h}{\sqrt{\epsilon'}}})$ time.
Hence, computing the $(1+\epsilon)$-approximate distance between $s$ and $t$ takes $O(n+\frac{h}{\sqrt{\epsilon'}}\lg{\frac{h}{\sqrt{\epsilon'}}})$ time.
For $\alpha = \frac{\epsilon}{2}$, the value of $\epsilon'$ is $O(\epsilon^2)$.
Hence, the result stated in the theorem statement.

For the plane sweep, leftmost and rightmost extreme vertices of convex polygons in $\cal{P}$ are found in $O(n)$ time.
There are $O(h)$ line segments in $\cal{T}$, cardinality of $\Omega$ is $O(h)$, and $O(h)$ line segment-obstacle pairs (respectively from $\cal{T}$ and $\cal{P}$) that intersect.
The number of event points due to the endpoints in sets $L, R$, and the endpoints of line segments in $\cal{T}$ is $O(h)$.
If $l$ and $P$ become non-adjacent along the sweep line, deleting their point of intersection from $Q$ is charged to the event that caused them non-adjacent.
The sweep line status is updated if any of these $O(h)$ number of event points occur.
Analogous to the analysis provided for line segment intersection \cite{books/compgeom/deberg2008}, our plane sweep algorithm takes $O(n+h\lg{h})$ time.  

Due to Dobkin et~al.~\cite{journals/tcs/DobkinK83}, determining whether a line segment $l$ in $SP_\Omega$ intersects with an obstacle $P$ takes $O(\lg{n})$ time, 
The preprocessing structures corresponding to \cite{journals/tcs/DobkinK83} take $O(n)$ space and they are constructed in $O(n)$ time.
Further, replacing every line segment between points of intersection with their respective geodesic shortest paths along the boundaries of obstacles together take $O(n)$ time.
\end{proof}

Note that the proof of the above theorem requires us to set the value of $\alpha$ to $\frac{\epsilon}{2}$.

\section{Approximate shortest path amid simple polygons}
\label{sect:polydom}

In this section, we extend the approximation method from previous sections to the case of simple (not necessarily convex) polygons.
This is accomplished by first decomposing $\cal{F}(\cal{P})$ into a set of corridors, funnels, hourglasses, and junctions \cite{conf/stoc/Kapoor99,conf/socg/Kapoor88,journals/dcg/KapoorMM97}. 
In the following, we describe these geometric structures, and then we detail our algorithm. 

For convenience, we assume a bounding box encloses the polygonal domain $\calP$.
In the following, we describe a coarser decomposition of $\cal{F(P)}$ as compared to the triangulation of $\cal{F(P)}$. 
In specific, this decomposition is used in our algorithm to achieve efficiency.
Let $\Tri(\calF)$ denote a triangulation of $\calF(\calP)$. 
The line segments of $\Tri(\calF)$ that are not the edges of obstacles in $\calP$ are referred to as {\it diagonals}.
Let $G(\calF)$ denote the dual graph of $\Tri(\calF)$, i.e., each node of $G(\calF)$ corresponds to a triangle of $\Tri(\calF)$ and each edge connects two nodes corresponding to two triangles sharing a diagonal of $\Tri(\calF)$.
Based on $G(\calF)$, we compute a planar 3-regular graph, denoted by $G_3$ (the degree of every node in $G_3$ is three), possibly with loops and multi-edges, as follows. 
First, we remove each degree-one node from $G(\calF)$ along with its incident edge; repeat this process until no degree-one node remains in the graph. 
Second, remove every degree-two node from $G(\calF)$ and replace its two incident edges by a single edge; repeat this process until no degree-two node remains. 
The resultant graph $G_3$ is planar, which has $O(h)$ faces, nodes, and edges.
Every node of $G_3$ corresponds to a triangle in $\Tri(\calF)$, called a {\it junction triangle}. 
The removal of all junction triangles results in $O(h)$ {\it corridors}.
The points $s$ and $t$ between which a shortest path needs to be computed are placed in their own degenerate single point corridors.    
The boundary of each corridor $C$ consists of four parts (see Fig.~\ref{fig:corridor}): 
(1) A boundary portion of an obstacle $P_i\in \calP$, from a point $a$ to a point $b$; 
(2) a diagonal of a junction triangle from $b$ to a point $e$ on an obstacle $P_j\in \calP$ ($P_i=P_j$ is possible); 
(3) a boundary portion of the obstacle $P_j$ from $e$ to a point $f$; 
(4) a diagonal of a junction triangle from $f$ to $a$.
The corridor $C$ is a simple polygon.
Let $\tau(a,b)$ (resp., $\tau(e,f)$) be the Euclidean shortest path from $a$ to $b$ (resp., $e$ to $f$) in $C$. 
The region $H_C$ bounded by $\tau(a,b), \tau(e,f)$, $\overline{be}$, and $\overline{fa}$ is called an {\it hourglass}, which is {\it open} if $\tau(a,b)\cap \tau(e,f)=\emptyset$ and {\it closed} otherwise. 
(Refer Fig.~\ref{fig:corridor}.)
If $H_C$ is open, then both $\tau(a,b)$ and $\tau(e,f)$ are convex polygonal chains and are called the {\it sides} of $H_C$; otherwise, $H_C$ consists of two {\it funnels} and a path $\tau_C=\tau(a,b)\cap \tau(e,f)$ joining the two apexes of the two funnels, and $\tau_C$ is called the {\it corridor path} of $C$.
Let $x$ and $y$ be the endpoints of $\pi_C$. 
Also, let $x$ be at a shorter distance from $b$ as compared to $y$.
The paths $\tau(b,x), \tau(e,x), \tau(a,y)$, and $\tau(f,y)$ are termed {\it sides of funnels} of hourglass $H_C$.
We note that these paths are indeed convex polygonal chains.
\begin{wrapfigure}{r}{0.6\textwidth}
\begin{minipage}[t]{\linewidth}
\vspace{-10pt}
\begin{center}
\includegraphics[totalheight=1.2in]{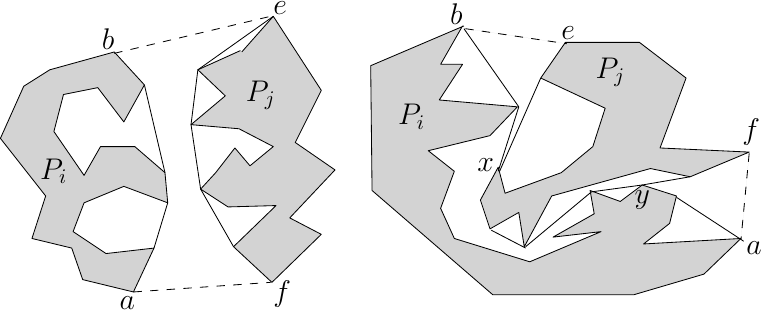}
\centering
\caption{\footnotesize Illustrating an open hourglass (left) and a
closed hourglass (right) with a corridor path connecting the apexes 
$x$ and $y$ of the two funnels. The dashed segments are diagonals.}
\vspace{-2pt}
\label{fig:corridor}
\end{center}
\end{minipage}
\end{wrapfigure}
The apieces $x$ and $y$ together is termed a {\it apex pair} of hourglass $H_C$.
Further, the shortest path between $x$ and $y$ along the boundary of $H_C$ is the {\it corridor path between apexes} of $H_C$.

We first give an overview of our algorithm for simple polygonal obstacles.
A sketch of $\cal P$ comprising of a sequence of convex polygonal (core-)chains is computed.
Each such core-chain either corresponds to an approximation of a side of an open hourglass or a side of a funnel.
If a simple polygon does not participate in any closed corridor, these polygonal chains together form a core-polygon.
Similar to the convex polygon case, each such polygonal chain is partitioned into patches.
Using these chains, we compute a spanner $G$.
In addition, the following set of edges are included in $G$: for every closed hourglass $H_C$ and for each obstacle $P$ that participates in $H_C$, an edge representing the unique shortest path between the two apieces of $H_C$ (as detailed below).
After we compute a shortest path $p$ between $s$ and $t$ in the spanner, for every edge $e(r',r'') \in p$, if $e$ is an edge that corresponds to the closed corridor path then we replace $e$ with a shortest path (sequence of edges) between $r'$ and $r''$ in $\cal{F(P)}$.
The resultant path is the output of our algorithm.
The scheme designed in Agarwal et~al.~\cite{conf/soda/AgarwalSY09} does not appear to extend easily to the case of simple polygons as they use the critical step of computing partitioning planes between pairs of convex polygonal obstacles from $\cal{P}$.

For every obstacle $P_j \in \cal{P}$, let ${\cal R}_j$ be the union of the following: 
(i) the set comprising of open hourglass sides whose endpoints are incident to $P_j$,
and
(ii) the set comprising of sections of funnel sides whose non-apex endpoints incident to $P_i$.
Note that the elements of sets in (i) and (ii) are polygonal convex chains.
For every $R \in {\cal R}_j$, similar to the case of convex polygonal obstacles, we partition $R$ into patches and the set comprising of the endpoints of these patches is the coreset of $R$.
(For details, refer to Section~\ref{sect:r2nopreproc}.)
For every $R \in {\cal R}_j$, the {\it core-chain of $R$} is obtained by joining every two successive vertices that belong to the coreset of $R$ with a line segment while traversing the boundary of $R$.
We construct a spanner $G(V, E)$ that correspond to core-chains of $\cal{P}$ using $CVD$s.
For every admissible cone $C_p$ at every vertex $p$ of every core-chain, we consider $C_p$ only if $C_p$ has an intersection with $\cal{F(P)}$.
While noting that Clarkson's method extends to core-chains defined as above, the shortest path determination algorithm for simple polygons is the same as for the convex polygons described in the previous section except for the following.
For each apex pair $v'$-$v''$, an edge $e$ is introduced into $G$ between the vertices of $G$ that correspond to $v'$ and $v''$  with the weight of $e$ equal to the geodesic distance between $v'$ and $v''$ in the closed hourglass.
For a shortest path $p$ between any two nodes of $G$, for every edge $e \in p$ if both the endpoints of $e$ correspond to an apex pair $a'$-$a''$ then we replace $p$ with the shortest path between $a'$ and $a''$ so that that path contains the corridor path of that closed hourglass;
otherwise, as in Lemma~\ref{lem:sketchpt}, we replace the line segment $l$ correspond to $e$ with the sections of $l$ together with the geodesic paths along the boundaries of patches that $l$ intersects.
Thus a shortest path between $s$ and $t$ in the spanner $G$ is transformed to a path in the $\mathcal{F(P)}$.
In addition, since the distance along the path that contains the corridor path between every pair of apexes is made as the weight of its corresponding edge in the spanner, and due to Lemma~\ref{lem:coneangle}, the distance along the transformed path is a $(1+\alpha\epsilon)$-approximation to the distance between $s$ and $t$ amid obstacles in $\calP$.

\begin{lemma}
\label{lem:noncvxpathconv}
For a set ${\cal P}$ of $h$ pairwise disjoint simple polygons in $\mathbb{R}^2$ and two points $s, t \in \cal{F(P)}$, the sketch of $\mathcal{P}$ with cardinality $O(\frac{h}{\sqrt{\alpha\epsilon}})$ suffices to compute a $(1+\alpha\epsilon)$-approximate shortest path between $s$ and $t$ in $\mathcal{F(P)}$.
\end{lemma}

Computing hourglasses of $\cal{F(P)}$ using \cite{conf/stoc/Kapoor99,conf/socg/Kapoor88,journals/dcg/KapoorMM97} and determining the core-chains together takes $O(n+h(\lg{h})^{1+\delta}+h\lg{n})$ time (where $\delta$ is a small positive constant resulting from the triangulation of $\cal{F(P)}$ using the algorithm from \cite{journals/ijcga/Bar-YehudaC94}).
Extending the proof of Theorem~\ref{thm:spcvx} leads to the following.

\begin{theorem}
\label{thm:sppolydom}
Given a set $\mathcal{P}$ of pairwise disjoint simple polygonal obstacles, two points $s, t \in \mathcal{F(P)}$, and $\epsilon \in (0, 0.6]$, a $(1+\epsilon)$-approximate geodesic shortest path between $s$ and $t$ is computed in $O(n + h((\lg{n}) + (\lg{h})^{1+\delta} + (\frac{1}{\epsilon}\lg{\frac{h}{\epsilon}})))$ time.
Here, $\delta$ is a small positive constant (resulting from the time involved in triangulating $\cal{F(P)}$ using \cite{journals/ijcga/Bar-YehudaC94}).
\end{theorem}

Same as in Theorem~\ref{thm:spcvx}, the proof of this theorem also needs the value of $\alpha$ to be equal to $\frac{\epsilon}{2}$.

\section{Two-point approximate distance queries amid convex polygons}
\label{sect:r2preproc}

We preprocess the given set $\calP$ of convex polygons to output the approximate distance between any two query points located in $\cal{F(P)}$.
Like in the previous section, our preprocessing algorithm relies on \cite{conf/stoc/Clarkson87} and constructs a spanner $G$.
Our query algorithm constructs an auxiliary graph from $G$. 
We compute the approximate distance between the two query points using a shortest path finding algorithm in the auxiliary graph.

\subsection{Preprocessing}
\label{subsect:algodetailed}

The graph $G$ constructed as part of preprocessing in Section~\ref{subsect:gcon} is useful in finding an approximate Euclidean shortest path in $\cal{F(P)}$ between any two vertices in $\mathcal{P}$.
Instead of finding a shortest path between two query nodes in $G$, to improve the query time complexity, we compute a planar graph $G^{pl}(V, E^{pl})$ from $G(V, E)$ using the result from Chew~\cite{journals/jcss/Chew89}. 
Chew's algorithm finds a set $E^{pl} \subseteq E$ in $O(|V|\lg{|V|})$ time so that the distance between any two nodes of $G^{pl}$ is a $2$-approximation of the distance between the corresponding nodes in $G$.
We use the algorithm from Kawarabayashi et~al.~\cite{conf/icalp/KawarabayashiKS11} to efficiently answer $(1+\epsilon)$-approximate distance (length) queries in $G^{pl}$.
More specifically, \cite{conf/icalp/KawarabayashiKS11} takes $O(|V|(\lg{|V|})^2)$ time to construct a data structure of size $O(|V|)$ so that any distance query is answered in $O((\frac{\lg{|V|}}{\epsilon})^2)$ time.

\begin{lemma}
\label{lem:preprocapprx}
Let $G$ be the spanner computed for the polygonal domain $\Omega$ using the algorithm mentioned in Subsection~\ref{subsect:gcon}.
Let $s$ and $t$ be two points in ${\cal F(P)}$.
Let $G^{pl}$ be the planar graph constructed from $G$ using \cite{journals/jcss/Chew89}.
Further, let $dist_{K}(s, t)$ be the distance between $s$ and $t$ in $G^{pl}$ computed using the algorithm from \cite{conf/icalp/KawarabayashiKS11}. 
By choosing $\alpha = \frac{\epsilon}{12}$, $dist_{\cal P}(s, t) \le dist_K(s, t) \le (2+\epsilon) dist_{\cal P}(s, t)$.
\end{lemma}
\begin{proof}
From Lemma~\ref{lem:numptsketch}, we know that $dist_{\Omega}(s, t) \le dist_{\cal P}(s, t) \le (1+\alpha\epsilon) dist_{\Omega}(s, t)$.
From Lemma~\ref{lem:coneangle}, we know that $dist_{\Omega}(s,t) \le dist_{G}(s, t) \le (1+\sqrt{\alpha\epsilon})dist_{\Omega}(s, t)$.
Let $dist_{G^{pl}}(s, t)$ be the distance in $G^{pl}$ between nodes $s$ and $t$ of $G^{pl}$.
From \cite{journals/jcss/Chew89}, $dist_G(s, t) \le dist_{G^{pl}}(s, t) \le 2\hspace{0.02in}dist_G(s, t)$.
Further, as mentioned above, $dist_{G^{pl}}(s, t) \le dist_{K}(s, t)$ $\le (1+\alpha\epsilon)dist_{G^{pl}}(s, t)$.
As detailed in Lemma~\ref{lem:sketchpt}, algorithm transforms a shortest path between $s$ and $t$ in $K$ to a path $p$ in $\cal{F(P)}$.
Let $dist_{\calP}^{p}(s, t)$ be the distance along $p$.
From Lemma~\ref{lem:sketchpt}, $dist_{\calP}^{p}(s, t) \le (1+\alpha\epsilon)dist_K(s, t)$.
Hence, $dist_{\calP}^{p}(s, t)
    \le (1+\alpha\epsilon)dist_K(s, t)
    \le (1+\alpha\epsilon)^2dist_{G^{pl}}(s, t)
    \le 2(1+\alpha\epsilon)^2dist_G(s, t)
    \le 2(1+\alpha\epsilon)^2(1+\sqrt{\alpha\epsilon})dist_{\Omega}(s, t)
    \le 2(1+\alpha\epsilon)^2(1+\sqrt{\alpha\epsilon})dist_{\cal P}(s, t)
    $.
Since $p$ is a path in $\cal{F(P)}$, it is immediate to note that $dist_{\cal P}(s, t) \le dist_{\calP}^{p}(s, t)$.
Therefore, $dist_{\cal P}(s, t) \le dist_{\calP}^{p}(s, t) \le 2(1+\alpha\epsilon)^2(1+\sqrt{\alpha\epsilon}) dist_{\cal P}(s, t)$.
To achieve $(2+\epsilon)$-approximation, $(2)(1+\alpha\epsilon)^2(1+\sqrt{\alpha\epsilon})$ needs to be less than or equal to $(2+\epsilon)$.
For small values of $\epsilon$ ($\epsilon \in (0, 0.7]$), choosing $\alpha = \frac{\epsilon}{12}$ satisfies this inequality.
\end{proof}

We note that $\alpha\epsilon$ is $O(\epsilon^2)$.
We suppose that there are $O(\frac{1}{\epsilon})$ cones in $\cal{C}$, each cone with a cone angle $O(\epsilon)$.
It remains to describe data structures that need to be constructed during the preprocessing phase for obtaining the closest vertex of the query point $s$ (resp. $t$) in a given cone $C_s$ (resp. $C_t$).
To efficiently determine all these $O(\frac{1}{\epsilon})$ neighbors to $s$ and $t$ during query time, we construct a set of $O(\frac{1}{\epsilon})$ $CVD$s: for every $C \in \cal{C}$, one $CVD$ that corresponds to $C$.
The $CVD$s are constructed similarly to the algorithm given in Subsection~\ref{subsect:gcon}.

\begin{lemma}
The preprocessing phase takes $O(n+\frac{h}{\epsilon^2}(\lg{\frac{h}{\epsilon}})+\frac{h}{\epsilon}(\lg \frac{h}{\epsilon})^2)$ time.
The space complexity of the data structures constructed by the end of the preprocessing phase is $O(\frac{h}{\epsilon})$.
\end{lemma}
\begin{proof}
Computing the sketch $\Omega$ from the given $\calP$ takes $O(n+\frac{h}{\epsilon})$ time.
The number of cones in all the $CVD$s together is $O(\frac{h}{\epsilon})$.
It takes $O(\frac{1}{\epsilon} \frac{h}{\epsilon} \lg{\frac{h}{\epsilon''}})$ time to compute $G$ which include computing $CVD$s.
Due to \cite{journals/jcss/Chew89}, computing planar graph $G^{pl}$ with $O(\frac{h}{\epsilon})$ nodes takes $O(\frac{h}{\epsilon}\lg{\frac{h}{\epsilon}})$ time.
Computing space-efficient data structures using \cite{conf/icalp/KawarabayashiKS11} takes $O(\frac{h}{\epsilon''} (\lg{\frac{h}{\epsilon}})^2)$ time.
Hence, the preprocessing phase takes $O(n+\frac{h}{\epsilon''}\lg{\frac{h}{\epsilon}}+\frac{h}{\epsilon}((\lg{\frac{h}{\epsilon}})^2)$ time. 
Further, data structures constructed using \cite{conf/icalp/KawarabayashiKS11} by the end of preprocessing phase occupy $O(\frac{h}{\epsilon})$ space.
The Kirkpatrick's point location \cite{journals/siamcomp/Kirkpatrick83} data structures for planar point location take $O(\frac{h}{\epsilon})$ space.
\end{proof}

\subsection{Shortest distance query processing}
\label{subsect:queryalgo}

The query algorithm finds the obstacle-avoiding Euclidean shortest path distance between any two given points $s, t \in \mathcal{F(P)}$. 
We construct a graph $G_{st}$ from $G^{pl}$. 
(The graph $G^{pl}$ is as defined in Subsection~\ref{subsect:algodetailed}.)
For every $C \in \mathcal{C}$, if the point $s$ is located in the cell of a point $p$ of $CVD$ corresponding to $C$, then we introduce a node corresponding to $p$ into a set $V_s$.
(Essentially, $p$ is the closest visible point in cone $-C_s$ to point $s$.)
Analogously, we define the set $V_t$ of nodes for $t$ in $G_{st}$.
The node set of $G_{st}$ comprises of nodes in $V_s \cup V_t \cup \{s, t\}$.
The edges of this graph are of three kinds: $\{s\} \times V_s, V_s \times V_t$ and $\{t\} \times V_t$.
Since there are $O(\frac{1}{\epsilon})$ CVDs, the number of nodes and edges of $G_{st}$ are respectively $O(\frac{1}{\epsilon})$ and $O(\frac{1}{\epsilon^2})$.
For every edge $(s, s')$ (resp. $(t, t')$) with $s' \in V_s$ (resp. $t' \in V_t$), the weight of edge $(s, s')$ (resp. $(t, t')$) is the Euclidean distance between $s$ and $s'$ (resp. $t$ and $t'$). 
For every edge $(s', t')$ with $s' \in V_s$ and $t' \in V_t$, the weight of $(s', t')$ is the $(2+\epsilon)$-approximate distance between $s'$ and $t'$.
These weights are obtained from the data structures maintained as in \cite{conf/icalp/KawarabayashiKS11}.
We use Fredman-Tarjan algorithm \cite{journals/jacm/FredmanT87} to find a shortest path between $s$ and $t$ in $G_{st}$.
From the above, this distance is a $(2+\epsilon)$-approximate distance from $s$ to $t$ amid convex polygons in $\cal{P}$.

\begin{theorem}
\label{thm:distq}
Given a set $\mathcal{P}$ of $h$ pairwise disjoint convex polygonal obstacles in plane defined with $n$ vertices and $\epsilon \in (0, 0.6]$, the polygons in $\calP$ are preprocessed in $O(n+\frac{h}{\epsilon^2}(\lg{\frac{h}{\epsilon}})+\frac{h}{\epsilon}(\lg \frac{h}{\epsilon})^2)$ time to construct data structures of size $O(\frac{h}{\epsilon})$ for answering two point $(2+\epsilon)$-approximate distance query between any two given points belonging to $\cal{F(P)}$ in $O(\frac{1}{\epsilon^6}(\lg {\frac{h}{\epsilon}})^2)$ time.
\end{theorem}

\ignore {
\section{Conclusions}
\label{sect:conclu}

The three algorithms presented rely on computing a sketch of the polygonal domain by extracting coresets of the polygonal input obstacles and further computing core-polygons from these coresets.
Our approximation algorithm for answering the two-point distance queries amid convex polygons improves the $(6+\epsilon)$ stretch of an algorithm given for this problem in \cite{conf/soda/Chen95} to $(2+\epsilon)$ when the polygonal obstacles are convex. 
Significantly, in contrast to \cite{conf/soda/AgarwalSY09}, our algorithm for convex polygonal obstacles extends to the case of computing a $(1+\epsilon)$-approximate shortest path between two points when the polygonal obstacles in $\calP$ are not necessarily convex.
}

\subsection*{Acknowledgements}
R. Inkulu's research is supported by NBHM grant 248(17)2014-R\&D-II/1049 and SERB MATRICS grant MTR/2017/000474.

\bibliographystyle{plain}


\pagebreak

\pagestyle{empty}

\vspace{-1in}


\begin{landscape}

\small{}

\subsection*{Appendix}

\vspace{0.1in}

Comparison with previous results:

\vspace{0.1in}

\scriptsize{}

\begin{tabular}{|c|c|c|c|c|c|c|}
    \hline
    & Preprocessing time & Space & Query time & Time & Stretch & Comment \\
    \hline
Our results & - & - & - & $O(n + h((\lg{n}) + (\lg{h})^{1+\delta} + \frac{1}{\epsilon}\lg{\frac{h}{\epsilon}}))$ & $1+\epsilon$ & non-convex \\
    \hline
    & - & - & - & $O(n+ \frac{h}{\epsilon}\lg{\frac{h}{\epsilon}})$ & $1+\epsilon$ & convex  \\
    \hline
    & $O(n+\frac{h}{\epsilon^2}(\lg{\frac{h}{\epsilon}})+\frac{h}{\epsilon}(\lg\frac{h}{\epsilon})^2)$ & $O(\frac{h}{\epsilon})$ & $O(\frac{1}{\epsilon^6}(\lg{\frac{h}{\epsilon}})^2)$ & - & $2+\epsilon$ & convex \\
    \hline
Agarwal et~al.~\cite{conf/soda/AgarwalSY09} & - & - & - & $O(n+\frac{h}{\sqrt{\epsilon}}\lg(\frac{h}{\epsilon}))$ & $1+\epsilon$ & convex \\
    \hline
Chiang \& Mitchell~\cite{conf/soda/ChiangM99} & - & $O(n^{5+\epsilon})$ & $o(n)$ & - & optimal & non-convex \\
    \hline
    & - & $O(n^{5+10\delta+\epsilon})$ & $O(n^{1-\delta} \lg{n})$ & - & optimal & non-convex \\
    \hline
    & - & $O(n^{10}\lg{n})$ & $O((\lg{n})^2)$ & - & optimal & non-convex\\
    \hline
    & - & $O(n^{11})$ & $O(\lg{n})$ & - & optimal & non-convex \\
    \hline
    & - & $O(n^{5})$ & $O(\lg{n}+\min{h_s,h_t})$ & - & optimal & non-convex \\
    \hline
    & - & $O(n+h^{5})$ & $O(h\lg{n})$ & - & optimal & non-convex \\
    \hline
Chen~\cite{conf/soda/Chen95} & - & $O(n\lg{n} + \frac{n}{\epsilon})$ & $O(\frac{\lg{n}}{\epsilon}+\frac{1}{\epsilon^2})$ & - & $6+\epsilon$ & non-convex \\
    \hline
Arikati et~al.~\cite{conf/esa/ArikatiCCDSZ96} & $O(\frac{n^2}{\sqrt{r}})$ & $O(\frac{n^2}{\sqrt{r}})$ & $O(\lg{n} + \sqrt{r})$ & - & $\sqrt{2} + \epsilon$ & non-convex \\
    \hline
     & $O(n \lg{n})$ & $O(n)$ & $O(n)$ & - & $\sqrt{2} + \epsilon$ & non-convex \\
    \hline
     & $O(n^{3/2})$ & $O(n^{3/2})$ & $O(\lg{n})$ & - & $2\sqrt{2} + \epsilon$ & non-convex \\
    \hline
     & $O(\frac{n^{3/2}}{\sqrt{\lg{n}}})$ & $O(n\lg{n})$ & $O(\lg{n})$ & - & $3\sqrt{2} + \epsilon$ & non-convex \\
    \hline
\end{tabular}

\vspace{0.3in}

In Chiang and Mitchell~\cite{conf/soda/ChiangM99}, $h_s$ (resp. $h_t$) is the number of vertices visible from $s$ (resp. $t$).
Both the Chen~\cite{conf/soda/Chen95} and Arikati et~al.~\cite{conf/esa/ArikatiCCDSZ96} output a shortest path in additional $O(L)$ time, where $L$ is the number of edges of the output path.
The $r$ in Arikati et~al. [2] is an arbitrary integer such that $1 \le r \le n$.

\normalsize{}

\end{landscape}

\end{document}